\documentclass[a4paper,11pt]{article}
\usepackage[textwidth=15cm]{geometry}
\usepackage{graphicx,color}  
\usepackage{amsfonts,amssymb,amsthm,amsmath}

\usepackage[colorlinks,citecolor=red,urlcolor=blue,bookmarks=false,hypertexnames=true]{hyperref} 

\newtheorem{theorem}{Theorem}[section]
\newtheorem{proposition}{Proposition}[section]
\newtheorem{lemma}[theorem]{Lemma}
\newtheorem{definition}[theorem]{Definition}

\begin{document}
\title{
Hybrid Koopman $C^*$--formalism and the hybrid quantum-classical master equation
}

\author{C. Bouthelier-Madre $^{1,2,3}$, J. Clemente-Gallardo $^{1,2,3}$, \\ L. González-Bravo $^4$, and D. Martínez-Crespo $^{1,3}$}
\date{$^1$ Departamento de F\'{\i}sica Te\'orica, Universidad de Zaragoza,  Campus San Francisco, 50009 Zaragoza (Spain) \\
$^2$ Instituto de Biocomputaci{\'{o}}n y F{\'{\i}}sica de Sistemas Complejos (BIFI), Universidad de Zaragoza,  Edificio I+D, Mariano Esquillor s/n, 50018 Zaragoza (Spain) \\
$^3$ Centro de Astropart{\'{\i}}culas y F{\'{\i}}sica de Altas Energías (CAPA), Departamento de F\'{\i}sica Te\'orica, Universidad de Zaragoza,  Campus San Francisco, 50009 Zaragoza (Spain)\\
$^4$ Departamento de Matemáticas, Universidad Carlos III de Madrid,   Leganés (Spain)
}
\maketitle
\abstract{
Based on Koopman formalism for classical statistical mechanics, we propose a formalism to define hybrid quantum-classical dynamical systems by defining (outer) automorphisms of the $C^*$ algebra of hybrid operators and realizing them as linear transformations on the space of hybrid states. These hybrid states are represented as density matrices on the Hilbert space obtained from the hybrid $C^*$--algebra by the GNS formalism. We also classify all possible dynamical systems which are unitary and obtain the possible hybrid Hamiltonian operators.
}

 \section{Introduction}

 Hybrid quantum-classical systems are physical models of systems where quantum degrees of freedom interact with classical ones. Their most common application is to approximate full quantum models  to simplify them while keeping an accurate model of the most relevant degrees of freedom. The most paradigmatic example are molecular models where most of the degrees of freedom are treated classically while using a quantum system to model the behavior of the most external electrons which are responsible for the chemical properties of the system.  There are many different approaches to define hybrid dynamical systems for molecular systems (see \cite{Yonehara2012} for a nice review), some of them are based on hybrid dynamics on the space of hybrid states (\cite{zhuNonBornOppenheimerLiouvillevonNeumann2005,Jasper2005,jasperNonbornoppenheimerMolecularDynamics2006,alonsoStatisticsNoseFormalism2011,alonsoEhrenfestDynamicsPurity2012,alonsoEhrenfestStatisticalDynamics2018,Agostini2014}), others are algorithmic (see \cite{tullyMolecularDynamicsElectronic1990,tullyMixedQuantumclassicalDynamics1998,Tully1998}), others are obtained as suitable limit equations of the full-quantum dynamics (\cite{Prezhdo1997,kapralMixedQuantumclassicalDynamics1999,Kapral2015,nielsenNonadiabaticDynamicsMixed2000,nielsenStatisticalMechanicsQuantumclassical2001}). Hybrid dynamical models appear also in other contexts, as those considering the problem of measurements of quantum systems with classical devices \cite{sherryInteractionClassicalQuantum1979,Buric2012} and  other type of systems (see \cite{Hall2008,Peres2001,Elze2013,Diosi2014} and references therein). More recently, some new approaches based on Koopman's formalism for classical statistical mechanics have been introduced with several remarkable contributions \cite{Bondar2019,Gay-Balmaz2020,Gay-Balmaz2022}.

 Among the different dynamical models, one of the most common is the Ehrenfest model, which can be easily obtained from the original full quantum one within a semiclassical description \cite{bornemannQuantumclassicalMolecularDynamics1996}. Let us use it as a reference example to explain the whole framework. Ehrenfest equations are defined on the cartesian product of the classical and quantum phase spaces $M_C\times M_Q$, which represent the hybrid pure states. This dynamics can be proved to admit a Hamiltonian description  with a suitable hybrid Poisson bracket and a hybrid Hamiltonian function, which combines the classical and the quantum energies. This property can be used to define a dynamical statistical model with Ehrenfest equations as dynamics of the microstates \cite{alonsoStatisticsNoseFormalism2011}. In this statistical description, the state of the system is defined as a probability distribution on the hybrid phase-space, which follows a Liouville equation defined by a hybrid Hamiltonian function. Thus, the dynamics on the probability density (i.e., the state) is defined as the dual to the dynamics on the space of hybrid observables. This is a natural way to define a consistent statistical mechanical system leading to a well defined thermodynamics (see \cite{balescuStatisticalDyamicsMatter1997}).
  This statistical description has relevant applications (\cite{alonsoEhrenfestDynamicsPurity2012,alonsoEhrenfestStatisticalDynamics2018}) but also some limitations, as the difficulty to write an entropy function and the corresponding notion of canonical ensemble \cite{alonsoEntropyCanonicalEnsemble2020} or, in more general terms, of an equilibrium thermodynamics.  From the mathematical point of view, these limitations are associated with the incompatibility of the notion of hybrid state as a probability density on the phase space and the definition of hybrid entropy, which require an alternative notion of state.  Indeed, the usual choice to represent the state of the hybrid system is a family of quantum-density-matrix-like operators, parametrized by classical variables. This type of hybrid states has been used extensively since the early eighties (see \cite{aleksandrovStatisticalDynamicsSystem1981}) and it  is used in many of the references presented above. But the challenge is to define a consistent master equation for this object, satisfying that:
  \begin{itemize}
    \item it is inner in the space of hybrid states, i.e., the defining properties of the state such as positiveness, normalization, etc are preserved at all times; 
    \item and it is dual to the dynamics of the observables;
  \end{itemize}
  which are the conditions required by the usual construction in Statistical Mechanics  (see \cite{balescuStatisticalDyamicsMatter1997}). This is a complicated task that, to the best of our knowledge, has not been satisfactorily solved yet.
  
  Our aim in this paper is to provide a new solution for this problem by changing our approach: instead of direcly looking for possible master equations, we will  consider the dynamics on the space of hybrid observables in a different way. First, we will build a hybrid $C^*$--algebra defined as the tensor product of the classical and the quantum $C^*$--algebras of observables of each subsystem. This is the same algebra that we would obtain if we model the physical system as a pure quantum one, and then take the classical limit of one of its subsystems.  Then, we will borrow Koopman's idea to define dynamics on the space of observables by defining an outer automorphism of the $C^*$--algebra, and implement that condition on our hybrid observables.  In order to do that, we will use the GNS construction (\cite{gelfandImbeddingNormedRings1943,Segal1947,Segal1947a}) to define a suitable representation of the algebra on a hybrid Hilbert space. In this way, we can impose linearity to the dynamics in a natural way, and make the treatment much more simple.   Once the dynamics on the observables is defined, we will build a  consistent master equation defining the dynamics of the physical state as the dualized system.  In the general case, several conditions must be imposed for the dynamics to be defined on the space of states. In this work, we will consider only the simplest case of unitary dynamics, where these conditions are immediately satisfied, but many others are admissible and will be considered and analyzed in a forthcoming paper. 

The structure of the paper is as follows.  In Section \ref{sec:introduction} we will consider the mathematical tools which will be used in the rest of the paper. First, we will summarize Koopman's formalism for classical statistical mechanics. Then, we will briefly discuss  the notion of $C^*$--algebra, the GNS construction and the notion of state. Finally, we will present these  constructions for the simple cases of the set of observables of classical and quantum systems. With these tools, we will build in Section \ref{sec:hybridalgebra} the notion of hybrid $C^*$--algebra and will characterize their states and the corresponding GNS representation. In Section \ref{sec:hybrid_dynamics} we will use Koopman's idea to define outer automorphisms of the hybrid algebra and discuss all the conditions that the dualized dynamics on the space of hybrid states must satisfy. Then, we will identify the simplest of these dynamics which will define unitary transformations of the GNS-Hilbert space. Finally, in Section \ref{sec:conclusions} we will summarize our main conclusions and discuss the analysis of alternative solutions which will be presented in future papers.

\section{Mathematical preliminaries}
\label{sec:introduction}

\subsection{Koopman formalism for Classical Statistical Mechanics}
\label{sec:Koopman}
We will follow the reference \cite{Chruscinski2006} in this brief summary. 
Let us consider a classical statistical system defined on a symplectic manifold $(M, \omega)$ of dimension $n$, through a measure $\mu$ defined by a probability density $\rho:M\to \mathbb{R}$, satisfying
$$
\mu(M)=\int_M d\mu=\int_M \rho \, d \Omega=1, \qquad d \Omega=\omega^n. 
$$
Consider the Hilbert space $\mathcal{L}^2(M)$ defined by square-integrable functions $f:M\to \mathbb{C}$ with respect to the scalar product 
\begin{equation}
\label{eq:scalarproduct}
\langle f_1 , f_2\rangle=\int_M \bar f_1 f_2 d\, \Omega.
\end{equation}
If we consider thus the state $\psi_\rho\in (\mathcal{L}^2(M), d\Omega)$ satisfying the condition
\begin{equation}
\label{eq:psi}
\rho=\bar \psi_\rho \psi_\rho.
\end{equation}
This object will represent the state of our statistical system in terms of the Hilbert space $\mathcal{L}^2(M)$.

If we consider a Hamiltonian dynamical system $X_M$ on $M$, it defines a Liouville equation for the probability density $\rho$ in the formalism
\begin{equation}
\label{eq:Liouville}
\dot \rho= -X_H(\rho).
\end{equation}
The corresponding flow $F_t:M\to M$ is a symplectomorphism $F_t^*\omega=\omega$, and therefore it defines a unitary transformation on  $\mathcal{L}^2(M)$:
\begin{equation}
\label{eq:unitary}
U_t(f)=F_t^* f, \qquad \forall f\in C^\infty(M)\subset \mathcal{L}^2(M).
\end{equation}
The result follows immediately from Liouville theorem and the linearity of the pullback of any differentiable mapping.

Being a unitary transformation, Stone theorem ensures the existence of a self-adjoint operator $L$ satisfying
\begin{equation}
\label{eq:Stone}
U_t=e^{-iLt}.
\end{equation}
Being the infinitesimal generator of the Hamiltonian evolution {\color{black}$U_t$ defined above, $L$ must be determined by} the Hamiltonian vector field $X_H$, i.e.
\begin{equation}
\label{eq:liouvillian}
L \psi_\rho= -i X_H (\psi_\rho)=-i \left ( \frac{\partial H}{\partial q^k}\frac{\partial}{\partial p_k}- \frac{\partial H}{\partial p_k}\frac{\partial}{\partial q^k}   \right ) \psi_\rho,
\end{equation}
where $(q^k, p_k)$ represent a chart of Darboux coordinates on $M$. 
This unitary operator translates the dynamics to the functions $\psi_\rho$ which thus satisfy
\begin{equation}
\label{eq:schrodinger}
i\dot \psi_\rho=L \psi_\rho.
\end{equation}

We can see that the classical wave functions are functions of both positions and momenta, and that the infinitesimal generator of the unitary evolution is a first-order differential operator. It is immediate from these relations to re-obtain in this quantum language the Liouville continuity Equation \eqref{eq:Liouville} for $\rho=\bar \psi_\rho \psi_\rho$.

From the point of view of the operators, notice that $q$'s and $p$'s are analogous variables and therefore the  operators having those variables as spectrum ($\hat Q$ and $\hat P$) must behave as multiplicative operators and hence they must commute. By extension, we can conclude that the whole algebra of classical observables is commutative when realized as linear operators on the Hilbert space. Hence, if we consider the Heisenberg picture on that algebra, we conclude that it is not possible to define an evolution corresponding to the action of the commutator with a Hamiltonian contained in the algebra, i.e., we can not consider inner automorphisms.  On the other hand, if we consider the adjoint action of the  Hamiltonian operator  \eqref{eq:liouvillian}, we can define a non-trivial evolution on the space of linear operators of $\mathcal{L}^2(M)$. Furthermore, the Hamiltonian  \eqref{eq:liouvillian} is chosen in such a way that the commutative subalgebra corresponding to the classical operators is preserved by the evolution, i.e., the evolution defines an automorphism of the classical subalgebra. {\color{black}Some other dynamics can also be considered, with analogous results on the classical system but small differences in the quantum model (see \cite{Jauslin2010})}.

In conclusion, Koopman formalism defines a commutative algebra of operators representing the classical magnitudes, which is represented as a subalgebra of the total space of linear operators of the Hilbert space  $\mathcal{L}^2(M)$. Dynamics corresponds to an outer-automorphism of that subalgebra, generated by the Hamiltonian operator \eqref{eq:liouvillian}.

{\color{black}
\subsection{Operator $C^*$--algebras and the GNS construction}

}
We will now consider the set of observables of our theory in more detail. 
In order to do that we are going to consider the $C^*$ algebras containing the set of classical, quantum and hybrid operators. For details on these topics, see \cite{landsmanMathematicalTopicsClassical1998,Bratelli1987}. 

Given a $C^*$--algebra $\mathcal{A}$, we can consider its dual space $\mathcal{A}^*$ and use the norm on the algebra to define a norm on the dual space. Indeed, given $\omega\in \mathcal{A}^*$, we define its norm as 
\begin{equation}
\label{eq:norm_state}
\| \omega\|=\mathrm{sup} \{ \vert \omega(a) \vert , \, \|a\|=1 \}.
\end{equation}
The involution on $\mathcal{A}$ allows also to introduce a notion of {\color{black}positivity for linear functionals, and thus for states}. Thus, given $\omega\in \mathcal{A}^*$ we say that it is positive definite if
\begin{equation}
\label{eq:positiveness}
\omega( a^*a) \geq 0 \qquad \forall a\in \mathcal{A}.
\end{equation}
\begin{definition}
 A \textbf{state} of a $C^*$--algebra $\mathcal{A}$ is defined as a positive linear functional on $\mathcal{A}$ with norm equal to one. 
\end{definition}

Finally, GNS theorem (see \cite{gelfandImbeddingNormedRings1943,Segal1947,Segal1947a}) ensures that given a $C^*$--algebra $\mathcal{A}$ and a state $\omega$, we can always build a representation $\pi$ of $\mathcal{A}$ on the set of (bounded) linear operators on a Hilbert space $\mathcal{B}(\mathcal{H})$ where the state is associated to a cyclic vector of $\mathcal{H}$, whose orbit under the representation of $\mathcal{A}$ is dense in $\mathcal{H}$.  This result allows us to recover the description of the elements of any $C^*$ algebra as linear operators on a suitable Hilbert space. 

Let us see now how this concept allows us to recover Koopman's construction in a simple way, and introduce a treatment of hybrid quantum-classical systems in these algebraic terms.

\subsection{Examples: classical and quantum systems}
{\color{black} Let us consider two examples which are relevant for us: the set of observables of a classical system, and the set of observables of a quantum one. As the quantum case is simpler and will provide us with some valuable properties, we will consider it first.}

\subsubsection{The quantum case}
Again, for the sake of simplicity we will consider the algebra of bounded operators $\mathcal{B}(\mathcal{H})$ over a Hilbert space $\mathcal{H}$ with respect to composition; and  the hermitian adjoint as the involution $A^*= A^\dagger$. Furthermore, we consider the norm of an operator $A$

\begin{equation}
\label{eq:normQ}
\| A \|=\mathrm{sup}\{ \|A \psi \|, \quad \psi\in \mathcal{H}, \,\,  \|\psi \|=1  \}
\end{equation}

Again, it is a well known fact that this set {\color{black} which will be denoted in the following as} $\mathcal{A}_Q$ becomes thus a $C^*$--algebra. Such an algebra contains the  set of observables of a quantum system, which corresponds to the subset of self-adjoint operators:
\begin{equation}
\label{eq:self-adjoint}
L=\left \{ A\in \mathcal{A}_Q \vert \, A^\dagger=A \right \}.
\end{equation}

Finally, regarding the set of states of $\mathcal{A}_Q$, it is well known that the set of quantum states is in one-to-one correspondence with the set of density matrices $\mathcal{D}(\mathcal{H})$ on the Hilbert space $\mathcal{H}$: 
{\color{black}
\begin{equation}
\label{eq:densitymatrices}
\mathcal{D}(\mathcal{H})= \left \{ \hat \rho \in \mathcal{B}(\mathcal{H}) \ \vert \, {\color{black} \hat \rho}= \hat \rho^\dagger; \hat \rho>0;  \mathrm{Tr} \hat\rho=1   \right \}
\end{equation}
From Gleason theorem (see \cite{gleasonMeasuresClosedSubspaces1957}), we know that with this set we are considering all possible states of a quantum system whose state space corresponds to a Hilbert space $\mathcal{H}$:
\begin{theorem}[Gleason]
  Let $\mathcal{H}$ be a separable Hilbert space with dimension greater than 2 and let $\mu$ be a measure defined on the closed subspaces of $\mathcal{H}$. Then, there exists a positive semi-definite self adjoint operator $\hat \rho$ of the trace class, satisfying that for any closed subspace $A\subset \mathcal{H}$,
  $$
\mu(A)=\mathrm{Tr} (\hat \rho P_A),
  $$
  where $P_A:\mathcal{H}\to \mathcal{H}$ is the orthogonal projection of the Hilbert space on $A$.
\end{theorem}

This result justifies the use of density matrices to represent the states of any quantum system for it encodes completely the probabilistic nature of the state (remember that the physical magnitudes, being represented by self-adjoint operators, can be written, by means of their spectral decompositions, as real linear combination of those orthogonal projectors).

Going back to the representation of the $C^*$--algebra, it is also well known that the }GNS construction defines an irreducible representation of $\mathcal{A}_Q$ if we select a pure state, and a reducible one if the density matrix is a mixed state {\color{black}(see \cite{Evans1998})}. The corresponding GNS representation $\pi_Q:\mathcal{A}_Q\to \mathcal{B}(\mathcal{H}_Q)$ is the natural one.

\subsubsection{The classical case}
For the sake of simplicity we will consider as an example {\color{black}to model the set of physical observables of a classical system} the set of compactly supported complex functions on a manifold $M_C$, $C_c(M_C, \mathbb{C})${\color{black}, although other possibilities may also be considered with analogous properties}. 
In order to endow this set with a $C^*$--algebra structure we will use:
\begin{itemize}
    \item the pointwise algebra $\cdot_C$
    \item the complex conjugation as involution $f^*(x)=\bar f(x)$,
    \item and the supremum norm $\| f\|=\mathrm{sup} \{ {\color{black} | f(x) |}\;| \; x\in M_C\}$
\end{itemize} 
It is immediate to verify that the set $\mathcal{A}_C=C_c(M_C, \mathbb{C})$ becomes thus a Banach algebra. This algebra can be considered to contain the (bounded) physical observables of a classical system (since the condition on compactness is a technical restriction which does not introduce serious physical limitations).

Finally, let us consider the states of $\mathcal{A}_C$.  From the definition, they must correspond to the positive-definite elements of the dual space having norm equal to 1. From the Riesz-Markov representation theorem {\color{black}(see, for instance, \cite{ReedSimon1981})}, we know that the set of states on $\mathcal{A}_C$ coincides with the set of {\color{black}(Radon)} measures on the domain of the classical functions, i.e., given $\omega\in \mathcal{A}_C^*$, there exists a measure $\mu$ satisfying  
\begin{equation}
\label{eq:riesz-markov}
\omega(a)=\int_{M_C} d\mu \, a, \qquad \forall a\in \mathcal{A}_C.
\end{equation}
If we consider a reference measure on $M_C$, as it might be the symplectic phase-space volume $d\Omega$, for most states (i.e., excluding particular cases as the Dirac delta) we can compute the corresponding Radon-Nikodym derivative and obtain thus a density $F_C:M_C\to \mathbb{R}^+$ satisfying 
\begin{equation}
\label{eq:density}
\omega(f)=\int_{M_C} d\Omega \, F_C a, \qquad \forall a\in \mathcal{A}_C.
\end{equation} 
From the physical point of view, states on $\mathcal{A}_C$ correspond then to statistical ensembles on the classical phase space with probability density $F_C$. {\color{black}This is the starting point of Koopman original proposal, a certain classical probability density. We can reproduce Koopman construction by considering the representation, via GNS construction, of the system  on a  Hilbert space.  It is well known, that} the GNS construction for $\mathcal{A}_C$ with state \eqref{eq:density} defines the representation $\pi_C:\mathcal{A}_C\to \mathcal{B}(\mathcal{H}_C)$ as a set of multiplicative operators acting on the Hilbert space $\mathcal{H}_C=\mathcal{L}^2(M_C, d\mu)$, as we saw in Section  \eqref{sec:Koopman}.

{\color{black}Once on the Hilbert space $\mathcal{H}_C$, we know from Gleason theorem that the state can be written as an element of the set of density matrices $\mathcal{D}(\mathcal{H}_C)$.} We can write, in a simple way, the expression of the density matrix on $\mathcal{H}_C$ in terms of the original classical density: 
{\color{black}
\begin{proposition}
\label{prop:classical}
Consider a classical state $\omega$ of the classical $C^*$--algebra $\mathcal{A}_C$, defined by a probability density function $F_C$ with respect to a measure $d\Omega$
$$
\omega=F_C d\Omega,
$$
where $d\Omega$ defines the cyclic vector of the GNS representation of $\mathcal{A}_C$.
 Then, the expression of the density matrix associated with $\omega$ by the GNS representation can be written as:
\begin{equation}
    \label{eq:classical_DM}
    \hat \rho_C=\int_{M_C} d\Omega(\xi) \int_{M_C} d\Omega(\xi ')\sqrt{F_C(\xi) F_C(\xi')} |\xi\rangle \langle \xi' |. 
    \end{equation}
\end{proposition}
\begin{proof}
  Indeed, if we take a function $a(\xi)\in \mathcal{A}_C$, and consider the multiplicative operator $\pi_C(a)=a(\xi)$, it follows
  $$
  \mathrm{Tr}\left ( \hat \rho_C \pi_C(a)  \right )= \int_{M_C} d\Omega(\xi) \int_{M_C} d\Omega(\xi ')\sqrt{F_C(\xi) F_C(\xi')} \mathrm{Tr} \left ( |\xi\rangle \langle \xi' | \pi_C(a)  \right ).
  $$
  
  As the classical algebra acts as multiplicative operators on $\mathcal{H}_C$, 
  
  \begin{equation}
  \label{eq:classical_cond}
  \mathrm{Tr} \left ( |\xi\rangle \langle \xi' | \pi_C(a) \right )=a(\xi) \delta(\xi-\xi'),
  \end{equation}
 we obtain that
  \begin{eqnarray}
  \mathrm{Tr}\left ( \hat \rho_C \pi_C(a)  \right )=& \int_{M_C} d\Omega(\xi) \int_{M_C} d\Omega(\xi ')\sqrt{F_C(\xi) F_C(\xi')} a(\xi) \delta(\xi-\xi')=  \nonumber \\ & \int_{M_C}d\Omega(\xi) F_C(\xi) a(\xi)=\omega(a),\qquad \forall a\in \mathcal{A}_C. 
  \end{eqnarray}

\end{proof}

}
\section{The hybrid $C^*$--algebra}
\label{sec:hybridalgebra}

Let us now consider the algebra containing the observables of a hybrid quantum-classical system. As in general the hybrid model is obtained as a suitable partial classical limit of a full-quantum model, a natural candidate corresponds to the tensor product of the two  $C^*$--algebras above, i.e. $\mathcal{A}_H=\mathcal{A}_C\otimes \mathcal{A}_Q$. {\color{black}As we are mostly focused on the application to hybrid quantum-classical physical systems, in the following we will consider only the tensor product of the two examples introduced above, even if many of our conclusions may be of interest for the product of arbitrary commutative and non-commutative $C^*$--algebras.}
{\color{black}
\subsection{The $C^*$--algebra structure}}
The hybrid product is defined in terms of the classical and the quantum products on separable states as:
\begin{equation}
\label{eq:hybridassociative}
(a\otimes A)\cdot_H(b\otimes B):=(a\cdot_C b)\otimes (A\cdot_Q B), \qquad \forall a,b\in \mathcal{A}_C, A, B \in \mathcal{A}_Q.
\end{equation}

{\color{black} Let us consider now the involution. }On the set of elements of the form
\begin{equation}
\label{eq:hybrid_element}
f=\sum_k \gamma_k a_k\otimes A_k, \quad \gamma_k\in \mathbb{C}, a_k\in \mathcal{A}_C, A_k\in \mathcal{A}_Q,
\end{equation}
we can consider the  operation \eqref{eq:hybridassociative}. This makes it an algebra which we will denote as $\mathcal{A}_H$.  On that algebra we can consider the involution 
\begin{equation}
\label{eq:involution}
f^*=\sum_k \bar \gamma_k a_k^*\otimes A_k^\dagger,
\end{equation}
where $a_k^*$ and $A_k^\dagger$ represent the classical and quantum involutions respectively. Clearly, this makes $\mathcal{A}_H$ an involutive algebra.

Regarding the definition of a norm, a few comments are in order (for a more detailed explanation see, for instance, \cite{Bruckler1999}).  In principle, there are different possible norms to be defined on the algebraic tensor product $\mathcal{A}_C\otimes \mathcal{A}_Q$ to make the set a $C^*$--algebra. But as we are interested in the GNS construction, the most natural candidate seems to be the \textit{spatial norm} defined by the inclusion of $\mathcal{B}(\mathcal{H}_C)\otimes \mathcal{B}(\mathcal{H}_Q)$  in $\mathcal{B}(\mathcal{H}_C\otimes \mathcal{H}_Q)$ and the definition of a representation 
\begin{equation}
\label{eq:GNS_h}
\pi_H=\pi_C\otimes \pi_Q,
\end{equation}
{\color{black} with} the norm
\begin{equation}
\label{eq:normH}
\| f \|=\| \pi_H(f) \|_{\mathcal{B}(\mathcal{H}_C\otimes \mathcal{H}_Q)}. 
\end{equation}
The construction was introduced by T. Turumaru in \cite{Turumaru1953} and does not depend on the particular representations of the factors.  Despite the diversity of possible norms in the general case, as $\mathcal{A}_C$ is a commutative algebra and hence a nuclear one, it is possible to prove that the $C^*$ norm on $\mathcal{A}_C\otimes \mathcal{A}_Q$ is unique (see \cite{Bruckler1999}). Hence, we will keep this construction above as the constitutive definition of the hybrid $C^*$--algebra structure for $\mathcal{A}_H$.

\subsection{Hybrid states}

As we saw above for a general $C^*$--algebra, hybrid states must be positive-definite elements of $\mathcal{A}_H^*$ with norm equal to 1. Obviously, the tensor product of a classical state and a quantum one satisfies these requirements. Hence, we may think in an example of hybrid state as the product of a classical measure $d\mu$  on $M_C$  and a quantum density matrix $\hat \rho_Q$. This is a particular case of  the representation of hybrid states used in the Literature (see \cite{aleksandrovStatisticalDynamicsSystem1981,Buric2012,buricUnifiedTreatmentGeometric2013,alonsoEhrenfestDynamicsPurity2012,alonsoEhrenfestStatisticalDynamics2018,alonsoEntropyCanonicalEnsemble2020,Bondar2019} and references therein).  If we use the GNS representation $\pi_H:\mathcal{A}_H\to \mathcal{B}(\mathcal{H}_C\otimes \mathcal{H}_Q)$ to write them as states on $\mathcal{H}_C\otimes \mathcal{H}_Q$, they become a tensor product of density matrices $\hat \rho_C\otimes  \hat \rho_Q$, i.e., this example turns out to be what in the quantum systems {\color{black} literature} is called a simply separable state.  {\color{black}From a physical point of view, its most remarkable property is the lack of correlations between the two subsystems, classical and quantum.}

More general states of $\mathcal{A}_H$ can be written as in the previous references as a family of  quantum operators parametrized by classical variables $\hat \rho(\xi)$, satisfying the normalization conditions
\begin{equation}
\label{eq:rho_xi}
\int_{M_C}d\mu(\xi) \mathrm{Tr}\hat \rho(\xi)=1.
\end{equation}
The action on the elements of $\mathcal{A}_H$ is written as 
\begin{equation}
\label{eq:stateH}
\langle f \rangle=\int_{M_C}d\mu(\xi) \sum_k \gamma_k a_k(\xi) \mathrm{Tr} (\hat \rho (\xi)A_k), 
\end{equation}
for an element $f\in \mathcal{A}_H$ written as Equation \eqref{eq:hybrid_element}. This is the usual representation of hybrid states in the Literature. {\color{black}As we argued above,  the definition of a suitable master equation of hybrid dynamics for this type of system is still an open problem. Next Section is entirely devoted to the definition of different possible solutions for it.

In order to obtain a few useful properties of these hybrid states $\hat \rho(\xi)$ let us consider their representation as elements of $\mathcal{D}(\mathcal{H}_C\otimes \mathcal{H}_Q)$, i.e., as density matrices. 
Notice that as the hybrid state defines a measure on $\mathcal{A}_H$, from Gleason theorem \cite{gleasonMeasuresClosedSubspaces1957}, there must exist a density matrix $\hat \rho_H$ to represent the state of the algebra $\pi_H(\mathcal{A}_H)$ on $\mathcal{H}_C\otimes \mathcal{H}_Q$. The states $\hat \rho_H$ must safisfy thus:
\begin{equation}
\label{eq:hybridstates_rho}
\langle f \rangle=\int_{M_C}d\mu(\xi)\mathrm{Tr}(\hat \rho(\xi) f)=\mathrm{Tr} (\hat \rho_H \pi_H(f)), \quad \forall f\in \mathcal{A}_H.
\end{equation}

 Our main interest in this representation as $\hat \rho_H$ is that, as we will see below, it} will allow us to write a well defined master equation which captures the hybrid dynamics.  The definition of the dynamics is {\color{black}a very difficult task} if the state is of the form \eqref{eq:hybridstates_rho} because of the non-linearity of the classical {\color{black} subsystem dynamics} (see \cite{Buric2013a}), but when realized at the level of the Hilbert space $\mathcal{H}_C\otimes \mathcal{H}_Q$ and as external to the algebra {\color{black}$\pi_H(\mathcal{A}_H)$}, it is possible to write it in a simple way. {\color{black} This is why we are interested in this generalization of Koopman classical construction.
 
 Inspired by Proposition \ref{prop:classical}, we are going to consider states on $\mathcal{H}_C\otimes \mathcal{H}_Q$ of the form
 \begin{equation}
  \label{eq:rho_H}
  \hat \rho_H=\sum\limits_{{mm^\prime}}\int_{M_C}\int_{M_C}d\Omega(\xi) d\Omega (\xi^\prime)\rho_{mm^\prime}(\xi,\xi^\prime)\vert \xi,m\rangle\langle\xi^\prime,m^\prime\mid,
  \end{equation}
  where $\rho_{{mm^\prime}}(\xi,\xi^\prime):=\sqrt{\langle m\mid\hat\rho(\xi)\mid m^\prime\rangle\langle m\mid\hat\rho(\xi^\prime)\mid m^\prime\rangle}$ and $\{ |m\rangle \}$ is a basis for $\mathcal{H}_Q$, which will be assumed to be discrete, for simplicity, (although this is not relevant). For those systems, it is immediate to prove:
  \begin{lemma}
    If we consider the marginalized state defined by the trace over $\mathcal{H}_C$ of $\hat \rho_H$ we obtain the same state over $\mathcal{H}_Q$ defined marginalizing the state $\hat \rho(\xi)$ of $\mathcal{A}_H$, i.e.
    \begin{equation}
      \mathrm{Tr}_C(\hat\rho_H)=\sum\limits_{{mm^\prime}}\int_{M_C} d\Omega(\xi) \langle m\mid\hat\rho(\xi)\mid m^\prime\rangle\mid m\rangle\langle m^\prime\mid=\int_{M_C}d\Omega(\xi) \hat\rho(\xi),
      \end{equation}
where $\mathrm{Tr}_C$ stands for the partial trace over $\mathcal{H}_C$, i.e., with that operation we can represent the marginalization of the state $\hat \rho_H$. More generally, it can be easily shown that:
      \begin{equation}
      \label{Partialclassicaltracek}
          \mathrm{Tr}_C(\hat\rho_H^k)=\int_{M_C} d\Omega(\xi) \hat\rho(\xi)^k\;,
      \end{equation}        
  \end{lemma}
  \begin{proof}
    The proof is completely analogous to the one used in the proof of Proposition \ref{prop:classical} and based on the fact that the trace corresponds to a $\delta(\xi-\xi^\prime)$ which produces the expressions above. 
  \end{proof}

With this result, we can conclude:

\begin{proposition}
\label{prop:states}
  The state $\hat \rho_H$ given by Eq. \eqref{eq:rho_H} is the density matrix  on $\mathcal{H}_C\otimes \mathcal{H}_Q$ corresponding to the hybrid state $\hat \rho(\xi)$.
\end{proposition}
\begin{proof}
 Again, the proof is analogous to the classical case.  If we consider separable hybrid observables we can write that
 \begin{eqnarray}
  \label{eq:relation}
  &\omega(a\otimes A)=\mathrm{Tr}(\hat\rho_H\pi_H(a\otimes A))= \nonumber 
  \\
  &\int_{M_C\times M_C}d\Omega_C(\xi) d\Omega_C(\xi^\prime)\mathrm{Tr}\left(\sum\limits_{{mm^\prime}}\rho_{{mm^\prime}}(\xi,\xi^\prime)\mid \xi,m\rangle\langle \xi^\prime,m^\prime\mid\pi_H(a\otimes A)\right)= \nonumber  \\
  &\int_{M_C}d\Omega_C(\xi) \mathrm{Tr}_Q\left(\sum\limits_{{mm^\prime}}\rho_{{mm^\prime}}(\xi,\xi)\mid m\rangle\langle m^\prime\mid (a(\xi)\otimes A)\right)= \nonumber \\
  &\int_{M_C}d\Omega_C(\xi) \mathrm{Tr}_Q\left(\hat\rho(\xi) (a(\xi)\otimes A)\right)\;,
  \end{eqnarray}
  where we wrote that $\pi_C(a)=a(\xi)$,  is a multiplicative operator on $\mathcal{H}_C$, and we used the previous lemma for the classical trace. $\mathrm{Tr}_Q$ represents the partial trace over $\mathcal{H}_Q$, and we use that $\mathrm{Tr}=\mathrm{Tr}_C\mathrm{Tr}_Q$. The proof for general elements of $\mathcal{A}_H$  of the form \eqref{eq:hybrid_element} is immediate.
  
\end{proof}

\subsection{Hybrid entropy function}
 }
An important application of this result is the possibility to relate the von Neumann entropy associated with the state $\hat \rho_H$ and the hybrid entropy function introduced in \cite{alonsoEntropyCanonicalEnsemble2020}.  {\color{black} Indeed, our group introduced a hybrid entropy function for states of the form $\hat \rho( \xi)$, based on the analysis of mutually exclusive hybrid events, which reads:
\begin{equation}
\label{eq:hybridentropy}
S_H[\hat \rho(\xi)]=-\int_{M_C}d\mu(\xi) \mathrm{Tr} \left ( \hat \rho(\xi) \log \hat \rho(\xi)  \right ).
\end{equation}
Based on this function, we were also able to identify a candidate for hybrid canonical ensemble, using the MaxEnt formalism. For this state to constitute a valid candidate for a thermodynamical equilibrium ensemble, it is necessary to identify a valid dynamics for the hybrid system, having the MaxEnt solution as a stable equilibrium point. Identifying such a hybrid dynamics is the main motivation for this work. As we argued above, searching for such a dynamics on the set of hybrid states of the form $\hat \rho(\xi)$ is a difficult task, since the classical subsystem makes the dynamics nonlinear.  Our proposal in this work is to generalize Koopman construction, define the hybrid system on a Hilbert space $\mathcal{H}_C\otimes \mathcal{H}_Q$, and search for possible hybrid dynamics on the set $\mathcal{D}(\mathcal{H}_C\otimes \mathcal{H}_Q)$.

In order to do that, our first task is to prove that we can recover write the hybrid entropy function in terms of the hybrid state $\hat \rho_H\in\mathcal{D}(\mathcal{H}_C\otimes \mathcal{H}_Q)$. A natural choice is to consider the well-known von Neumann entropy 
$$
S_{vN}[\hat \rho_H]=-\mathrm{Tr} (\hat \rho_H \mathrm{log}\hat \rho_H).
$$

\begin{lemma}
  Let $\hat \rho(\xi)$ be a state for the hybrid algebra $\mathcal{A}_H$ and $\hat \rho_H \in\mathcal{D}(\mathcal{H}_C\otimes \mathcal{H}_Q)$ the state of the algebra $\pi_H(\mathcal{A}_H)\in \mathcal{B}(\mathcal{H}_C\otimes \mathcal{H}_Q)$. Then, 
  $$
S_{vN}[\hat \rho_H]=S_H[\hat \rho(\xi)].
  $$
\end{lemma}
}
\begin{proof}

As the density matrix is self-adjoint and hence diagonalizable, we can consider its spectral decomposition, where the spectrum is nowhere negative. From the properties of the trace and the definition of the logarithm as a series, we can work directly on the spectrum, and obtain:
\begin{eqnarray}
\label{eq:entropy_H}
S_H[\hat \rho_H]=& -\mathrm{Tr}\left(\hat\rho_H\sum_{n=1}^\infty(-1)^{n-1}\frac{(\hat\rho_H-\mathbb{I})^n}{n}\right)= \nonumber\\
& \mathrm{Tr}\left(\sum_{n=1}^\infty\frac{(-1)^{n-1}}{n}\sum_{k=0}^n\binom{n}{k}{\hat\rho}^{k+1}_H(-1)^{n-k}\right)= \nonumber \\
&\mathrm{Tr}_Q\left(\sum_{n=1}^\infty\frac{(-1)^{n-1}}{n}\sum_{k=0}^n\binom{n}{k}\mathrm{Tr}_C({\hat\rho}^{k+1}_H)(-1)^{n-k}\right) 
\end{eqnarray}
Making use of \eqref{Partialclassicaltracek}, we can substitute the partial trace over the classical part by an integral over phase space, and thus:
\begin{equation}
\label{eq:entropy_H_hybrid}
-S_H[\hat \rho_H]=\int_{M_C}d\Omega(\xi)\mathrm{Tr}_Q\left(\sum_{n=1}^\infty\frac{(-1)^{n-1}}{n}\sum_{k=0}^n\binom{n}{k}{\hat\rho(\xi)}^{k+1}(-1)^{n-k}\right)=-S_H[\hat\rho(\xi)],
\end{equation}
which is the hybrid entropy function introduced in \cite{alonsoEntropyCanonicalEnsemble2020}.  

\end{proof}

{\color{black} This}
is an important result from the physical point of view, since having an entropy function we can apply MaxEnt formalism to identify the state which maximizes the entropy while keeping the average energy fixed. This allows us to write the canonical ensemble of a hybrid system in this {\color{black}generalized} Koopman formalism in a straightforward way. From our analysis above, we conclude that it can be obtained from the result identified in \cite{alonsoEntropyCanonicalEnsemble2020} {\color{black} using Proposition \ref{prop:states}}. It is important to remark, though, that the MaxEnt computation for the case of $\hat \rho_H$ must be done with respect to operator $\pi_H(f_H)\in \pi_H(\mathcal{A}_H)$ associated with the hybrid energy function $f_H\in \mathcal{A}_H$, and not with respect to a Hamiltonian operator defining the dynamics as we saw in the Koopman classical formalism. Such an operator may generate the dynamics, but it does not have a physical meaning as energy of the system since it does not belong to the hybrid algebra. We shall discuss these issues in some detail in the following section.

It is also important to emphasize that the MaxEnt argument is completely independent of the dynamics of the microstates, but it only offers possible candidates to equilibrium ensembles. For these MaxEnt solutions to define actual thermodynamical ensembles, it is necessary to prove that the dynamics of the microstates preserves the solution obtained.  Hence, the existence of candidates to equilibrium ensembles defines constraints  to the possible dynamics that can be considered on the hybrid system, as we discuss in the next section.

\section{Hybrid dynamical systems}
\label{sec:hybrid_dynamics}
\subsection{Classical and quantum dynamics}
As we saw in Section \eqref{sec:Koopman}, Koopman's original construction proved that Liouville evolution equation of a classical statistical system on a phase space $M_C$ can be realized as a unitary one-parameter group of transformations   on a suitably defined Hilbert space $\mathcal{H}_C$. In order to do that, a Hamiltonian operator {\color{black}on $\mathcal{H}_C$ which does not belong to the classical (commutative) $C^*$-algebra which contains the set of physical magnitudes, is required} (see \cite{koopmanHamiltonianSystemsTransformation1931}). From that unitary group of transformations we can define equivalent dynamical systems on the Hilbert space $\mathcal{H}_C$ (via the Schrödinger equation), on the commutative subalgebra of $\mathcal{B}(\mathcal{H}_C)$ defined by the representation $\pi_C:\mathcal{A}_C\to \mathcal{B}(\mathcal{H}_C)$ (via the Heisenberg equation), or on the set of density matrices $\mathcal{D}(\mathcal{H}_C)$ (via von Neumann equation). All three systems are physically equivalent to the solutions of the Liouville equation on the set of statistical states on $M_C$, or the equivalent Hamiltonian evolution on the Poisson algebra of classical observables. Notice, though, that the dynamics must be defined on $\mathcal{A}_C$ as an outer automorphism since it is commutative and has a trivial Lie structure. We must add an additional Poisson tensor to the set of functions in $\mathcal{A}_C$ to be able to define the dynamics at the level of the algebra. And therefore the resulting dynamical system defines outer-automorphisms of the $C^*$--algebra. Nonetheless, this construction is compatible with the $C^*$--algebra description of the classical system (in terms of $\mathcal{A}_C$) presented above and the corresponding GNS representation $\pi_C:\mathcal{A}_C\to \mathcal{B}(\mathcal{H}_C)$. Thus, when written in the Hilbert space language, dynamics takes the usual form, but with a Hamiltonian which does not belong to $\pi_C(\mathcal{A}_C)$. This corresponds to the usual Koopman's construction. Notice that, in the representation process, the non-linear classical Liouville equation (with a Hamiltonian function in the classical algebra) becomes the linear Heisenberg equation (with a Hamiltonian which does not belong to the subalgebra $\pi_C(\mathcal{A}_C)$) {\color{black} or the linear von Neumann equation for the corresponding density matrix}. In a certain way,  nonlinearities are ``smoothed" by the representation. 

The quantum case is different since the definition of a dynamical system on $\mathcal{A}_Q$  is straightforward using its Lie canonical structure and the corresponding Heisenberg equation. Choosing a self-adjoint Hamiltonian in $\mathcal{A}_Q$ (even if non-bounded, in general), we can define a unitary evolution which can be implemented either at the level of the Hilbert space, at the level of the operator algebra, or at the level of the states, exactly as in the classical case. But finite dynamical transformations correspond to  bounded operators and hence inner automorphisms of the Lie structure of $\mathcal{A}_Q$. 

\subsection{Hybrid dynamics}
\subsubsection{General considerations}
Let us consider now the hybrid case. {\color{black} As we explained above, our final goal is to consider possible dynamics of hybrid states in order to identify those which are useful to model statistical quantum-classical systems. Among those, we will have to check whether or not they have the hybrid canonical ensemble identified with the MaxEnt formalism as a stable fixed point. If we succeed, we would have found an efficient way to model statistical hybrid system at finite temperature, which is a very relevant situation for molecular simulations. In this paper we will just consider the first problem: how to identify possible hybrid dynamics. We will define the problem and classify the solutions which are unitary, as in the case of classical Koopman dynamics. More general solutions and the evaluation on the hybrid canonical ensemble will be considered in future papers. }

We will consider directly the GNS representation of the system and therefore a system characterized by some density matrix $\hat \rho_H\in \mathcal{D}(\mathcal{H}_C\otimes \mathcal{H}_Q)$ and {\color{black} the algebra $\pi_H(\mathcal{A}_H)\subset \mathcal{B}(\mathcal{H}_C\otimes \mathcal{H}_Q)$}. The problem of defining a dynamical system directly at the level of the algebra $\mathcal{A}_H$ or its dual, which has received much more attention in the Literature, and its relation with our solution in this paper will be considered in a forthcoming publication. {\color{black} In this paper we will consider the definition of dynamics only at the level of the Hilbert space $\mathcal{H}_C\otimes \mathcal{H}_Q$. }

From what we learned in the classical and quantum case, it is clear that we are supposed to build an automorphism of the $C^*$ algebra $\mathcal{A}_H$, or, equivalently, of its image $\pi_H(\mathcal{A}_H)\subset \mathcal{B}(\mathcal{H}_C\otimes \mathcal{H}_Q)$. {\color{black}This automorphism generalizes Koopman construction to the hybrid case.} If the dynamics must act on the classical degrees of freedom in a non-trivial way, it should contain external elements to $\mathcal{A}_C$ (and hence to $\mathcal{A}_H$ when multiplied by quantum operators) and define an external automorphism of the subalgebra $\pi_H(\mathcal{A}_H)$. In this way we define a dynamical system on the space of hybrid physical magnitudes, which preserves the set, unlike what Ehrenfest dynamics was seen to do on $\mathcal{A}_H$ (\cite{Buric2013a}). 

As a generalization of the classical and quantum cases above, we can ask the dynamics to fulfill the following requirement: it may be defined on the whole $\mathcal{B}(\mathcal{H}_C\otimes \mathcal{H}_Q)$ but it must preserve the subalgebra $\pi_H(\mathcal{A}_H)$. For the sake of simplicity, we will consider only linear systems.  Therefore, we will consider a dynamical system of the form\begin{equation}
    \label{eq:automorphism}
     \frac{ d\pi_H(f) (t)}{dt}=\mathcal{L}\pi_H(f) (t), \qquad  \forall f \in \mathcal{A}_H,
    \end{equation}
where $\mathcal{L}$ represents a linear super-operator on $\mathcal{B}(\mathcal{H}_C\otimes \mathcal{H}_Q)$ and satisfies $\mathcal{L}( \pi_C (\mathcal{A}_H))\subset \pi_C (\mathcal{A}_H)$. In that case, we will write the master equation as an equation on the set of states as:
    \begin{equation}
    \label{eq:mastereq_gen}
     \frac{ d \hat \rho_H (t)}{dt}= \mathcal{L}^\dagger \hat \rho_H(t),
    \end{equation}
where $\mathcal{L}^\dagger$ represents the adjoint operator to $\mathcal{L}$. For the sake of simplicity, we assume that the operator $\mathcal{L}$  generates a bounded operator $e^{\mathcal{L}t}$. Nonetheless, this just  defines a flow on the dual space to $\mathcal{B}(\mathcal{H}_C \otimes \mathcal{H}_Q)$. Furthermore, {\color{black}we know that the set of density matrices $\mathcal{D}(H_C\otimes H_Q)$ corresponds to just a subset of that dual space (since we impose the conditions on the trace and positivity) and  we want the dynamics to preserve that subset}.  These conditions must also be imposed {\color{black} as restrictions on the dynamical system $\mathcal{L}$.}  Besides, after our analysis on the hybrid entropy and the equilibrium ensembles, it also makes sense to impose some entropic restrictions to the possible dynamical systems. If we want the hybrid dynamics to represent a microstate dynamics for a physical system, the value of von Neumann entropy on $\hat \rho_H(t)$  must be constant in time (if the hybrid system is isolated) or increase (if it is not isolated). 

A particularly simple case corresponds to the case of unitary dynamics, where there exists a Hamiltonian operator $\hat H$ which allows to write the operator $\mathcal{L}$ as its adjoint action, i.e.:
\begin{equation}
\label{eq:Heisenberg_hybrid}
\frac{ d\pi_H(f) (t)}{dt}=-i \left (\pi_H(f) (t) \hat H -\hat H\pi_H(f) (t)  \right ), \quad f\in \mathcal{A}_H.
\end{equation}
In this case, the dual equation (the corresponding von Neumann equation) 
\begin{equation}
\label{eq:vonNeumann_hybrid}
\frac{d\hat \rho_H(t)}{dt}=-i \left (\hat H \hat \rho_H(t)- \hat  \rho_H(t) \hat H \right ), \quad \hat \rho_H \in \mathcal{D}(\mathcal{H}_C\otimes \mathcal{H}_Q)
\end{equation}
is known to preserve the set of density states, since it defines the orbits of the coadjoint action of the evolution operator. 
Furthermore, in this case it is immediate that a unitary transformation preserves the von Neumann entropy of the state, since it preserves its spectrum. We will analyze the properties of these dynamical systems in the following section.

\subsubsection{{\color{black}Conditions on the automorphism of $\pi_H(\mathcal{A}_H)$ I: unitary dynamics}}

For the sake of simplicity, let us first consider the case of a unitary transformation, i.e., we consider as dynamical equation the adjoint action of a certain Hamiltonian operator $\hat H$ {\color{black}(as in Equation \eqref{eq:Heisenberg_hybrid})}. As we did above, we can write the Hamiltonian operator without loss of generality as the sum of three terms:
\begin{equation}
\label{eq:Hybrid_Ham}
\hat H= \hat H_C\otimes \mathbb{I}_Q+\mathbb{I}_C\otimes \hat H_Q+\hat H_{CQ}, 
\end{equation}
where again $\hat H_C$ represents the energy associated to the classical degrees of freedom, $\hat H_Q$ the energy of the quantum ones, while $\hat H_{CQ}$ represents the coupling between them. 

Written in this form, it is simple to study how to define a (unitary) automorphism of the image of the hybrid algebra $\mathcal{A}_H$. We know that the dependence in the classical degrees of freedom can not be only on those of $\mathcal{A}_C$, but also on those of $\mathcal{B} (\mathcal{H}_C)$ which do not belong to $\pi_C(\mathcal{A}_C)$. As $\pi_C(\mathcal{A}_C)$ is the (commutative) subalgebra of multiplicative operators on $\mathcal{L}^2(M_C, d\mu$), we can look for the external operators among those corresponding to the derivation operators, i.e., those representing the quantization of the variables conjugated to those of $M_C$. From Koopman's construction, we know that the quantization of those variables define suitable operators to define the proper unitary dynamics.   For the sake of simplicity, let us consider the case ${\color{black}M_C}=\mathbb{R}^{2n}$,  and let us denote as $\Pi_{q^k}$ and $\Pi_{p_j}$ the conjugated coordinates to $q^k$ and $p_j$ (required to define the quantization from $T^*{\color{black}M_C}$ on $\mathcal{H}_C=\mathcal{L}^2({\color{black}M_C})$). Let us extend the representation mapping $\pi_C$ to include these functions and define the corresponding operators $\pi_C(\Pi_{q^k})$ and $\pi_C(\Pi_{p_j})$. If we ask the adjoint action of Hamiltonian \eqref{eq:Hybrid_Ham} to preserve the subalgebra $\pi_H(\mathcal{A}_H)$ we must impose the following condition
\begin{equation}
\label{eq:condition_Ham}
[\pi_H (\mathcal{A}_H), \hat H] \subset \pi_H (\mathcal{A}_H). 
\end{equation}
As the action is linear, it is sufficient to impose the condition on separable operators of the form $\pi_C(a)\otimes \pi_Q(A)$; and ask the result to belong to the algebra. Also the coupling term in the Hamiltonian can be supposed to be a linear combination of separable operators on $\mathcal{H}_C$ and $\mathcal{H}_Q$ as $H_{CQ}=\sum_{jk} c_{jk} h^j_C\otimes h^k_Q$, for $c_{jk}\in \mathbb{R}$.  As the condition must hold for any hybrid observable, it is immediate to prove that:
\begin{itemize}
    \item the term $\hat H_C$ in the Hamiltonian must be linear in the operators $\pi_C( \Pi_{q^k})$ and $\pi_C( \Pi_{p_j})$ to recover the non-trivial classical dynamics on the classical degrees of freedom. Indeed, as $\pi_C$ and $\pi_Q$ are morphisms onto $\mathcal{B}(\mathcal{H}_C)$ and $\mathcal{B}(\mathcal{H}_Q)$ we can write
    \begin{equation}
    \label{eq:classical_term}
    [\pi_C(a)\otimes \pi_Q(A), \hat H_C\otimes \mathbb{I}_Q]=[\pi_C(a), \hat H_C]\otimes \pi_Q(A)
    \end{equation}
    The condition for $[\pi_C(a), \hat H_C]$ to belong to $\pi_C(\mathcal{A}_C)$ implies that $\hat H_C$ must be linear in $\pi_C (\Pi_{q^k})$ and $\pi_C (\Pi_{p_j})$. This is analogous to the Koopman's case, which coincides with this one if we fix to zero all other terms in the Hamiltonian.
    \item Purely quantum terms do not introduce any constraint on the dynamics.
    \item Finally, the coupling term $\hat H_{CQ}$ can not depend on the operators $\pi_C (\Pi_{q^k})$ and $\pi_C (\Pi_{p_j})$. Indeed, if we consider analytical elements in the hybrid algebra, 
  
    \begin{eqnarray}
    \label{eq:couplig_term}
 &   [\pi_C(a)\otimes \pi_Q(A), \hat H_{CQ}]=\sum_{jk}c_{jk}[\pi_C(a)\otimes \pi_Q(A), h_C^j\otimes h_Q^k]=  \nonumber \\
  & \sum_{jk}c_{jk} \left ([ \pi_C(a), h_C^j]\otimes \pi_Q(A) h^k_Q+\pi_C(a) h_C^j\otimes [\pi_Q(A), h_Q^k] \right ). 
    \end{eqnarray}
    
    For these terms to belong to $\pi_H(\mathcal{A_H})$, we obtain that
    
    $$[\pi_C(a), h^j_C] \in \mathcal{A}_C
     \qquad \pi_C(a) h^j_C\in \mathcal{A}_C; \qquad \forall \pi_C(a)\in \pi_C(\mathcal{A}_C), \quad \forall j.
    $$
     This only happens if 
    $$
    h_C^j\in \pi_C(\mathcal{A}_C), \qquad \forall j.
    $$
\end{itemize}
Hence, we have proved that:
\begin{theorem}
\label{thm:result}
    The only type of Hamiltonian operator of the form of Equation \eqref{eq:Hybrid_Ham} which generates a unitary dynamics on $\mathcal{B}(\mathcal{H}_C\otimes \mathcal{H}_Q)$  that defines an outer automorphism of the hybrid subalgebra $\pi_H(\mathcal{A}_H)$ has
     \begin{itemize}
     \item  a linear dependence on $\pi_C(\Pi_{q^k})$ and $\pi_C(\Pi_{p_j})$ in $\hat H_C$ as
     \begin{equation}
     \label{eq:Hc}
     \hat H_C=\sum_{kj}\left (\alpha_{k} \pi_C(\Pi_{q^k}) + \beta_j \pi_C(\Pi_{p_j}) \right ),
     \end{equation}
     where $\alpha_{k}, \beta_{j},\tilde H_C\in \pi_C(\mathcal{A}_C)$.  
     \item  The other coefficients must belong to the corresponding subalgebra, i.e., $\hat H_Q\in \mathcal{A}_Q, \hat H_{CQ}\in \pi_H(\mathcal{A}_H)$.
     \end{itemize}
\end{theorem}

The dynamics is non-trivial in the classical degrees of freedom only if $\alpha_{k}\neq 0$ or $\alpha_{j}\neq 0$. In order to recover the expression of Koopman's construction, we may write
\begin{equation}
\label{eq:classical_coefficients}
\alpha_{k}=\pi_C \left (\frac{\partial H_C(q,p)}{\partial p_k}\right ), \qquad  \beta_{j}=- \pi_C\left (\frac{\partial H_C(q,p)}{\partial q^j}\right ), 
\end{equation}
where the function $H_C(q,p)\in \mathcal{A}_C$ represents the energy of the classical degrees of freedom. With that choice, for a non-interacting system ($H_{CQ}=0$),  a separable state of the form $\hat \rho_C\otimes \hat \rho_Q$ would evolve separately since the evolution operator factorizes. Hence we our construction includes as limit cases the classical and the quantum cases. Any non-trivial interaction (i.e., $H_{CQ}\neq 0$) modifies both terms, and introduces a hybrid behavior.  Obviously, in any case, the evolution is Hamiltonian. {\color{black}A similar dynamics is also considered in \cite{Jauslin2010}, from a slightly different perspective.}

\subsubsection{The master equation}
 {\color{black}Let us consider now the corresponding} equation on the set of density matrices.  Notice that the immediate consequence of Theorem \eqref{thm:result} above is that if we consider the master equation for the state $\hat \rho_H$ defined by {\color{black} von Neumann equation for} the same Hamiltonian, we obtain: 
\begin{eqnarray}
\label{eq:Hybrid_masterEq}
i\frac{d\hat \rho_H}{dt}=&\sum_{kj} \left [\pi_C \left (\frac{\partial H_C(q,p)}{\partial q^j}\Pi_{p_j}-\frac{\partial H_C(q,p)}{\partial p_k} \Pi_{q^k}\right ) \otimes \mathbb{I}_Q,\hat \rho_H \right ]+ \nonumber \\ & 
 [\mathbb{I}_C\otimes \hat H_Q, \hat \rho_H]+  [\hat H_{CQ}, \hat \rho_H],
\end{eqnarray}
where $H(q,p)\in \mathcal{A}_C$, $\hat H_Q\in \pi_Q(\mathcal{A}_Q)$ and $\hat H_{CQ}\in \pi_H(\mathcal{A}_H)$ and we used that $\pi_C(\mathcal{A}_C)$ is commutative. By construction, this evolution takes the previous Heisenberg equation to the set of physical states, and being unitary, it defines the corresponding dual unitary action on that space. Hence, the solution must be well defined inside the space of hybrid states since it corresponds to the co-adjoint action of the unitary group. In this way, this construction allows us to write a well defined dynamics on the set of hybrid states.

It is important to notice that, despite the automorphism being external to $\pi_H (\mathcal{A}_H)$, it is completely determined by the hybrid energy via the previous expression. As the dependence must be linear in the momenta $\Pi_{q^k}, \Pi_{p_j}$, there is no freedom left in the external degrees to determine the Hamiltonian.  

On the other hand, at the formal level the dynamics can be integrated straightforwardly, since it is defined by the evolution operator generated by the Hamiltonian  $\hat H$. If we assume that the Hamiltonian does not depend on time, the evolution operator corresponds to 
\begin{equation}
\label{eq:evolution}
\hat U(t)=e^{-i t \hat H},
\end{equation}
which is a unitary operator on $\mathcal{H}_C\otimes \mathcal{H}_Q$, whose adjoint action on the set of linear operators preserves the subalgebra $\pi_H(\mathcal{A}_H)$. Therefore, its co-adjoint action must preserve the corresponding set of density operators. {\color{black}This case represents the simplest generalization of Koopman construction to the case of hybrid systems.}

\subsubsection{{\color{black} Conditions on the automorphism of $\pi_H(\mathcal{A}_H)$ II: arbitrary linear dynamics.}}
If we consider arbitrary {\color{black}(bounded)} linear dynamics of the form of Equation \eqref{eq:automorphism}, possibilities are much richer. In principle, as we are considering linear subspaces, the condition of being an automorphism for $\pi_H(\mathcal{A}_H)$ implies that the operator generating the automorphism satisfies
\begin{equation}
\label{eq:automorphism_cond}
e^{\mathcal{L}t}(\pi_H(\mathcal{A}_H))\subset \pi_H(\mathcal{A}_H).
\end{equation}
Furthermore, if the automorphism is not unitary, we need to impose also some conditions on the dual operator $\mathcal{L}^\dagger$ for the master equation to be well defined. In particular we must take into account:
\begin{itemize}
    \item the evolution must be tangent to the set of density matrices, i.e., define a curve in the space of self-adjoint positive-definite operators. {\color{black}As we can safely assume that the Hilbert space under consideration is separable, we can always choose a numerable  basis and turn this space into $\ell^2(\mathbb{N})$ under a unitary isomorphism. Then, the Cholesky decomposition indicates that it is always possible  to factor the positive semidefinite density matrix  $\rho=LL^\dagger$ with $L$ lower triangular with respect to the order induced in the basis by the  isomorphism (see \cite{Paulsen2016}). Hence}
     a simple way to implement this condition may be to impose that
    \begin{equation}
    \label{eq:positivity}
    e^{\mathcal{L}^\dagger t} (T_1^\dagger T_1)=T_2^\dagger T_2, 
    \end{equation}
    for $T_1, T_2$ {\color{black}arbitrary bounded invertible mappings}.
    \item Furthermore, the trace must be preserved:
    \begin{equation}
    \label{eq:trace}
    \mathrm{Tr} \hat \rho_H(t)=1, \quad \forall t\Rightarrow \mathrm{Tr} \left (\frac {d\hat \rho_H(t)}{dt} \right )=0.
    \end{equation}
 
 \item Finally, if the dynamics is considered as the dynamics of a thermodynamic system, we must ask it to be compatible with the corresponding physical constraints. Thus, von Neumann entropy must be preserved or increased along the solutions (depending on the type of system) and the ensembles which must be of equilibrium (as the hybrid canonical ensemble discussed above) must be fixed points of the dynamics.

\end{itemize}

In general, the set of admissible dynamical systems may be large. We will explore this set and its properties in a forthcoming paper.

\section{Conclusions and outlook}
\label{sec:conclusions}

In this paper we have introduced a new framework based on Koopman's construction to study hybrid quantum-classical dynamics. We first consider the set of hybrid operators $\mathcal{A}_H$ defined as the tensor product of the classical and the quantum $C^*$--algebras. The quantum $C^*$ algebra corresponds to the usual algebra of linear operators of a Hilbert space. As it is well known, it is non-commutative. The classical $C^*$--algebra is commutative and corresponds to the differentiable functions of the classical phase space. We can consider that it is the result of the classical limit of a quantum algebra.  From this point of view we can think on the hybrid algebra as a partial classical limit of an originally full-quantum one. 

We can use the usual GNS construction on the hybrid algebra to define a representation of the hybrid observables as a subalgebra of the (bounded) linear operators of a hybrid Hilbert space $\mathcal{H}_C\otimes \mathcal{H_Q}$. The corresponding hybrid states are then realized as regular density matrices on this Hilbert space. We have built explicitly the density matrices from the states of the hybrid algebra, and we have also proved that von Neumann entropy of those density matrices coincides with the notion of hybrid entropy introduced in \cite{alonsoEntropyCanonicalEnsemble2020} for the hybrid states. This allows us to import to this framework a relevant physical result as the hybrid canonical ensemble determined in the same paper.  Hybrid dynamics is introduced on the algebra as an outer-automorphism of the representation of $\mathcal{A}_H$ on $\mathcal{H}_C\otimes \mathcal{H_Q}$ or, equivalently, as the corresponding dual dynamics on the set of density matrices. We have characterized the conditions of those dynamics, and have classified all possible unitary evolutions satisfying them. Notice, though, that from a physical point of view, the resulting family of unitary dynamics is limited, since it can not include back-reaction, i.e., the classical subsystem follows a dynamics which does not see the quantum subsystem. A detailed analysis of more general solutions and their properties will be presented in a forthcoming paper. We will also consider the relation with the dynamics written directly at the level of the $C^*$--algebra and its states, which is the most frequent case in the literature. 

Notice that, despite the apparent paradox of treating hybrid classical-quantum systems in a Hilbert space language, it offers several advantages. First, complexity is significantly decreased from an original full-quantum model, since the Hilbert space $\mathcal{H}_C$ representing the classical degrees of freedom is simpler than the original Hilbert space containing all the degrees of freedom, and its operators form a commutative subalgebra. Second, Hilbert space language offers several tools to consider linear dynamics which become non-linear at the classical level (as Koopman dynamics shows). Finally, it also offers the possibility of studying classical-quantum correlations of the statistical ensembles in a simple way, {\color{black}considering} the hybrid density matrices. There are also some drawbacks, as the difficulties to consider pure statistical hybrid states (as it happens in the Koopman's pure classical case). Nonetheless, from a physical point of view, the most common applications of these models correspond to experimental situations where it is impossible to assign precise initial conditions to the physical particles (for instance in a molecular system) and then we consider that missing pure states is not a serious limitation.

\section*{Acknowledgements}
The authors acknowledge partial finantial support of  Grant PID2021-123251NB-I00 funded by MCIN/AEI/10.13039/501100011033  and by the European Union, and of Grant E48-23R funded by Government of Aragón. L. G-B acknowledges that this work has been partially supported by the Madrid Government (Comunidad de Madrid-Spain) under the Multiannual Agreement with UC3M in the line of “Research Funds for Beatriz Galindo Fellowships” (C\&QIG-BG-CM-UC3M), and in the context of the V PRICIT (Regional Programme of Research and Technological Innovation). C.B-M and D.M-C acknowledge financial support by Gobierno de Aragón through the
grants defined in ORDEN IIU/1408/2018 and ORDEN CUS/581/2020 respectively.


\end{document}